\RequirePackage{amsthm}
\documentclass[sn-mathphys,Numbered,pdflatex]{sn-jnl}

\usepackage{graphicx}%
\usepackage{multirow}%
\usepackage{amsmath,amssymb,amsfonts}%
\usepackage{amsthm}%
\usepackage{mathrsfs}%
\usepackage[title]{appendix}%
\usepackage{xcolor}%
\usepackage{textcomp}%
\usepackage{manyfoot}%
\usepackage{booktabs}%
\usepackage{algorithm}%
\usepackage{algorithmicx}%
\usepackage{algpseudocode}%
\usepackage{listings}%

\usepackage{overpic}

\theoremstyle{thmstyleone}%
\newtheorem{theorem}{Theorem}
\newtheorem{lemma}[theorem]{Lemma}%

\theoremstyle{thmstyletwo}%
\newtheorem{example}{Example}%

\theoremstyle{thmstylethree}%
\newtheorem{definition}{Definition}%

\usepackage{mathtools}
\DeclarePairedDelimiter\abs{\lvert}{\rvert}

\DeclarePairedDelimiter\ceil{\lceil}{\rceil}

\DeclarePairedDelimiter\parenv{\lparen}{\rparen}

\DeclarePairedDelimiter\set{\{}{\}}

\renewcommand{\leq}{\leqslant}

\renewcommand{\geq}{\geqslant}

\newcommand{\cT}{\mathcal{T}}

\newcommand{\N}{\mathbb{N}}

\newcommand{\Z}{\mathbb{Z}}

\newcommand{\eqdef}{\triangleq}
\newcommand{\der}{\Longrightarrow}
\newcommand{\vv}{\phantom{|}}
\newcommand{\dis}{\mathrm{dis}}
\DeclareMathOperator{\pos}{Pos}
\DeclareMathOperator{\midp}{Mid}
\DeclareMathOperator{\pref}{Pref}

\begin{document}

\title[On Duplication-Free Codes for Disjoint or Equal-Length Errors]{On Duplication-Free Codes for Disjoint or Equal-Length Errors}


\author*[1]{\fnm{Wenjun} \sur{Yu}}\email{wenjun@post.bgu.ac.il}

\author[1,2]{\fnm{Moshe} \sur{Schwartz}}\email{schwartz.moshe@mcmaster.ca}

\affil[1]{\orgdiv{School of Electrical and Computer Engineering}, \orgname{Ben-Gurion University of the Negev}, \orgaddress{\city{Beer Sheva}, \postcode{8410501}, \country{Israel}}}

\affil[2]{\orgdiv{Department of Electrical and Computer Engineering}, \orgname{McMaster University}, \orgaddress{\city{Hamilton}, \postcode{L8S 4K1}, \state{ON}, \country{Canada}}}


\abstract{
Motivated by applications in DNA storage, we study a setting in which strings are affected by tandem-duplication errors. In particular, we look at two settings: disjoint tandem-duplication errors, and equal-length tandem-duplication errors. We construct codes, with positive asymptotic rate, for the two settings, as well as for their combination. Our constructions are duplication-free codes, comprising codewords that do not contain tandem duplications of specific lengths. Additionally, our codes generalize previous constructions, containing them as special cases.
}




\maketitle

\section{Introduction}

Ever since the first experiments to show the possibility of storing information in DNA molecules~\cite{ChuGaoKos12,GolBerCheDesLepSipBri13}, and later in living organisms~\cite{ShiNivMacChu17}, the last decade has seen a flurry of activity in the area of coding for DNA storage. While storing information in DNA molecules outside of living organisms (in-vitro DNA storage) is somewhat denser, using DNA molecules within living organisms (in-vivo DNA storage) also enables in-vivo synthetic-biology methods and algorithms that need ``memory,'' allows watermarking genetically-modified organisms (GMOs), and enables labeling organisms in biological studies. Overall, it is estimated that DNA storage can be six orders of magnitude denser than electronic media.

In-vivo DNA storage schemes display a variety of errors. Some are common also in electronic storage and communication systems, such as substitution errors and insertions/deletions errors. However, \emph{duplication errors} are unique to DNA storage, and their study began with~\cite{JaiFarSchBru17a}.  Generally speaking, a duplication error copies a substring of the stored DNA sequence and places the copy somewhere else in the sequence. Several variations exist, motivated by biological mutation processes: tandem duplication (where the copy is placed next to its original location), interspersed duplication (where the copy is placed some distance from the original copy), reverse duplication (same as tandem duplication, but the new copy is placed in reverse order), and reverse-complement duplication (same as reverse duplication, but with the new copy being complemented as well). Error models may also differ by the possible lengths of the duplicated substrings, and the total number of duplication errors.

Interest in correcting duplication errors has expanded significantly. Codes correcting any number of tandem duplications were studied in~\cite{JaiFarSchBru17a,ZerEsmGul19,ZerEsmGul20}, reverse duplications in~\cite{ZerEsmGul20}, and reverse-complement duplications in~\cite{YohSch23}. Codes that correct only a fixed number of tandem duplications (sometimes just one) were studied in~\cite{Kov19,LenWacYaa19,GosPolVor23}, reverse duplications in~\cite{NguCaiSonImm22}, and reverse-complement duplications in~\cite{BenSch22}. Additionally, codes that correct mixtures of duplications, substitutions, as well as insertions and deletions, have been the focus of~\cite{TanYehSchFar20,TanFar21a,TanWanLouGabFar23}. We also mention related work on duplication errors, though not providing error-correcting codes, but rather studying the possible outcomes from multiple errors~\cite{FarSchBru16,JaiFarBru17,EliFarSchBru19,FarSchBru19,BenSch22,Eli24}.

In this paper, we focus on the setting where any number of tandem-duplication errors occur. This setting was the focus of~\cite{JaiFarSchBru17a,ZerEsmGul19}, however, apart from several ad-hoc constructions, the only parametric code construction was for the uniform tandem-duplication case, namely, the case where all the tandem duplications were of a fixed prescribed length $\ell\in\N$. In almost all cases, the constructed codes were based on strings containing no tandem duplications of certain lengths, called \emph{duplication-free codes}.

We extend the results of~\cite{JaiFarSchBru17a,ZerEsmGul19} by considering two new settings: disjoint tandem-duplication errors, and equal-length tandem-duplication errors. We provide parametric constructions for both cases, which are duplication-free codes with carefully chosen forbidden duplication lengths. In the latter case, our construction contains the parametric construction of~\cite{JaiFarSchBru17a} and one ad-hoc case as special cases. We also combine the two settings, constructing codes for disjoint equal-length tandem-duplication errors. All of our code constructions have positive asymptotic rate.

The paper is organized as follows. In Section~\ref{sec:prelim}, we present our notation, provide exact definitions for our settings, and place previous results in context. Next, in Section~\ref{sec:cons}, we give all our constructions and prove their correctness, as well as discuss how to decode, when possible efficiently. Finally, in Section~\ref{sec:conc} we summarize our results and mention a couple of open questions.

\section{Preliminaries}
\label{sec:prelim}

Throughout this paper, let $\Sigma$ denote a finite alphabet, the elements of which are called letters. We further assume $\abs{\Sigma}\geq 2$, i.e., the alphabet contains at least two letters. A string is a sequence of letters, $u=u_1 u_2 \dots u_n$, where $u_i\in\Sigma$ for all $i$. In that case we say the length of $u$ is $n$, and denote it by $\abs{u}=n$. The set of all strings of length $n$ over $\Sigma$ is denoted by $\Sigma^n$, and the set of all strings of finite length is denoted by $\Sigma^*$. The unique empty string (of length $0$) is denoted by $\varepsilon$. Note that $\varepsilon\in\Sigma^*$ always. The set of all strings of finite positive length is denoted by $\Sigma^+\eqdef \Sigma^*\setminus\set{\varepsilon}$. Given two strings, $u,v\in\Sigma^*$, we use $uv$ to denote their concatenation. If $t\in\N$ is a positive integer, then $u^t$ denotes a concatenation of $t$ copies of $u$.

We say $y\in\Sigma^\ell$ is an \emph{$\ell$-factor} (or an \emph{$\ell$-substring}) of $u$, if there exist $x,z\in\Sigma^*$ such that $u=xyz$. If the length of $y$ is unimportant we may just say that it is a factor of $u$. If $x=\varepsilon$, then we say that $y$ is an $\ell$-prefix of $u$. If $z=\varepsilon$, we say that $y$ is an $\ell$-suffix of $u$.

We shall sometimes find it useful to have a notation for the positions occupied by a certain factor. Thus, if $u=u_1 u_2 \dots u_n\in\Sigma^n$, with each $u_i\in\Sigma$, then we use brackets to denote the set of positions between them. More precisely, for all $1\leq i<j\leq n$ we define
\[
\pos( u_1 \dots u_{i-1} [ u_i \dots u_{j}] u_{j+1} \dots u_n) \eqdef \set*{i,i+1,\dots j}.
\]
We also define the \emph{midpoint} of the factor as
\[
\midp( u_1 \dots u_{i-1} [ u_i \dots u_{j}] u_{j+1} \dots u_n) \eqdef \frac{i+j}{2}.
\]
We note that when the length of the factor is odd, the midpoint is the index of the middle letter of the factor. When the length of the factor is even, it is not an integer, but rather the average of the indices of the two middle letters. Assume $i'<j'$ are positive integers. By abuse of notation we shall say
\[
\midp( u_1 \dots u_{i-1} [ u_i \dots u_{j}] u_{j+1} \dots u_n) \in \set*{i',i'+1,\dots, j'},
\]
if
\[
i' \leq \midp( u_1 \dots u_{i-1} [ u_i \dots u_{j}] u_{j+1} \dots u_n) \leq j'.
\]

We now introduce string-duplication systems. We generally follow the definitions and notation of~\cite{FarSchBru16}. The most basic component is a \emph{string-duplication rule}, which is nothing more than a function $T:\Sigma^*\to\Sigma^*$. A set of such rules is denoted by $\cT\subseteq {\Sigma^*}^{\Sigma^*}$. The choice of rules is mainly motivated by simulating biological processes that occur during DNA storage, such as tandem duplication, which we shall shortly define (see~\cite{FarSchBru16} for more examples).

Given a string $u\in\Sigma^*$, we say that $v\in\Sigma^*$ is a \emph{$t$-descendant} of $v$, if there exist $T_1,T_2,\dots,T_t\in\cT$, which are not necessarily distinct, such that $v=T_{t}(T_{t-1}(\dots T_1(u)\dots))$. We denote this relation by writing $u\der^t v$. If $t=1$ we simply write $u\der v$. If $t$ is unknown or unimportant, but finite, we write $u\der^* v$. For convenience, we define $\der^0$ to be the identity, namely, $u\der^0 u$. The set of all $t$-descendants of $u$ is denoted by
\[ D^t(u) \eqdef \set*{ v\in\Sigma^* ~:~ u\der^t v}.\]
In a similar fashion we define $D^*(u)$, which we call the \emph{descendant cone} of $u$. Thus, $D^*(u)$ is the reflexive transitive closure of $\cT$ acting on $u$.

In a setting such as in-vivo DNA storage, information may be stored as a DNA string $u\in\Sigma^n$. Biological processes may corrupt the stored information by applying a sequence of mutations on it from the set of all allowed mutations, $\cT$. Thus, when reading the information, instead of retrieving $u$, a corrupted version $v\in D^*(u)$ from the descendant cone of $u$ is obtained. A natural way of defining error-correcting codes for such a scenario was given in~\cite{JaiFarSchBru17a}. A \emph{duplication-correcting code} capable of correcting any number of duplication mutations is a subset $C\subseteq\Sigma^n$, such that for all distinct $u,u'\in C$ we have $D^*(u)\cap D^*(u')=\emptyset$. We say two distinct strings $v,v'\in\Sigma^*$ are \emph{confusable} if $D^*(v)\cap D^*(v')\neq\emptyset$. Thus, any two distinct codewords in the code $C$ are not confusable. We observe that the code implicitly depends on the possible duplications, $\cT$, which we sometimes emphasize by saying that $C$ is a duplication-correcting code \emph{with respect to $\cT$}.

With the code defined as above, we say the code has length $n$, and size $\abs{C}$. An important figure of merit for a code is its rate, which is defined by
\[ R(C) \eqdef \frac{1}{n}\log_{\abs{\Sigma}}\abs*{C}.\]
Given a sequence of codes, $\set{C_i}_{i\in\N}$, where $C_i$ has length $n_i$, and $n_{i+1}>n_i$ for all $i$, we are interested in the \emph{asymptotic rate} of the family, defined as $\limsup_{i\to\infty} R(C_i)$. We would like the asymptotic rate to be positive, and as high as possible, with the highest value being called the \emph{coding capacity} of the channel.

The object of interest in this paper is the \emph{tandem-duplication system}. The duplication rules for this system are all of the following form:
\begin{equation}
\label{eq:tandem}
T_{i,\ell}(x) \eqdef 
    uv^2 w \qquad \text{if $x=uvw$, $\abs{u}=i$, $\abs{v}=\ell$.}
\end{equation}
Thus, $T_{i,\ell}$ takes the $\ell$-factor of the input $x$, which is found after a prefix of length $i$, and duplicates that factor immediately after its original appearance in $x$. Note that $T_{i,\ell}$ is not defined on strings of length shorter than $\ell+i$. For a set of positive integers, $L\subseteq\N$, we then define
\[ \cT_L \eqdef \set*{T_{i,\ell} ~:~ i\geq 0, \ell\in L}.\]

Duplication-correcting codes capable of correcting an unbounded number of duplications have already been studied, and in most cases, the duplication rules are tandem duplications. The first case, considered in~\cite[Theorem 15]{JaiFarSchBru17a}, was $L=\set{\ell}$, $\ell\in\N$, a singleton. The code obtained was optimal (in size). Other cases considered were $L=\set{1,2}$ and $L=\set{1,2,3}$ (see~\cite[Theorems 27, 32]{JaiFarSchBru17a}), where for the case of $L=\set{1,2}$ the codes were optimal. These results were extended in~\cite[Theorem 1]{ZerEsmGul19} to $L=\set{2,3}$, in~\cite[Theorem 2]{ZerEsmGul19} to $L=\set{1,2,3,4}$, in~\cite[Theorem 4]{ZerEsmGul19} to $L=\set{1,\dots,5}$, and in~\cite[Theorem 5]{ZerEsmGul19} to $L=\set{1,\dots,\ell}$ for $5<\ell\leq 9$. All of these codes have positive asymptotic rate.

All the constructions from~\cite{JaiFarSchBru17a}, and one construction from~\cite{ZerEsmGul19}, used codes based on duplication-free codewords. Here, a string $u\in\Sigma^*$ is said to be \emph{$F$-duplication free} (also called \emph{irreducible}) if it does not contain a factor of the form $v^2$ where $\abs{v}\in F$, for some subset $F\subseteq\N$. We call $F$ the \emph{forbidden duplication-length set}. We can now formally define the codes we focus on in this paper:

\begin{definition}
\label{def:cf}
    Given $F\subseteq\N$ and $n\in\N$, the $F$-duplication-free code of length $n$ is defined as
    \[
    C_F \eqdef \set*{ x\in\Sigma^n ~:~ \nexists u,v,w\in \Sigma^* \text{ s.t. } x=uv^2 w, \abs*{v}\in F}.
    \]
\end{definition}

It is interesting to note that in the cases of $L=\set{\ell}$, $L=\set{1,2}$, $L=\set{1,2,3}$, and $L=\set{2,3}$, the constructions from~\cite{JaiFarSchBru17a,ZerEsmGul19} (up to a small modification that does not change the asymptotic rate\footnote{The codes $C_F$ in~\cite{JaiFarSchBru17a} contain all the duplication-free strings of length \emph{up to} $n$, with a proper padding to make them all of length $n$.}) have $F=L$. Namely, by forbidding certain duplication lengths from occurring in codewords, the resulting codes can correct any number of duplications of these lengths. As another side note, when $\abs{\Sigma}\geq 3$, the number of duplication-free strings, even for $F=\N$ (called \emph{square-free} strings), grows exponentially with $n$ (see~\cite[Theorem 6]{Ber05}), which guarantees a positive asymptotic rate for such codes.

\section{Code Construction}
\label{sec:cons}

In this section we shall construct codes that can correct an unbounded number of tandem duplications. The only known code construction that is parametric in its error-correction capabilities, is the duplication-free code, $C_F$, where $F=L=\set{\ell}$ is a singleton (see~\cite{JaiFarSchBru17a}). In what follows, we provide two main constructions with parametric error-correction capabilities. Both constructions choose $F$ and $L$, while also imposing a restriction on the occurrence of duplication errors. In the first scenario, the duplication errors are disjoint, and in the second scenario, the duplication errors are all of equal length. Our second construction generalizes $C_{\set{\ell}}$ and contains it as a special case. We conclude by combining the two scenarios.

We begin by introducing relevant relations between the locations of multiple duplications in a single string.

\begin{definition}
\label{def:disjoint}
Let $x\in\Sigma^*$ be a string that contains two tandem duplications, namely
\[ x=uvvw=u'v'v'w',\]
for some $u,u',w,w'\in\Sigma^*$ and $v,v'\in\Sigma^+$. We say the two duplications are \emph{disjoint} if
\[
\pos(u[vv]w)\cap\pos(u'[v'v']w')=\emptyset.
\]
\end{definition}

We also want to define a related concept of obtaining a string from another through disjoint duplication errors.

\begin{definition}
\label{def:disjointerror}
Let $x\in\Sigma^*$ be a string. We say \emph{$z\in\Sigma^*$ is obtained from $x$ using $t\in\N$ disjoint tandem duplications} if we can write
\begin{align*}
x &= x_1 v_1 x_2 v_2 x_3 \dots x_t v_t x_{t+1},\\
z &= x_1 v_1^2 x_2 v_2^2 x_3 \dots x_t v_t^2 x_{t+1},\\
\end{align*}
where $x_i\in\Sigma^*$ for all $1\leq i\leq t+1$, and $v_i\in\Sigma^+$ for all $1\leq i\leq t$.
\end{definition}

Obviously, if $z$ is obtained from $x$ using $t$ disjoint tandem duplications (Definition~\ref{def:disjointerror}, then $z$ contains $t$ pair-wise disjoint tandem duplications (Definition~\ref{def:disjoint}). However, the reverse direction does not necessarily hold. 

Another important concept is presented in the following definition:

\begin{definition}
\label{def:midcover}
Let $x\in\Sigma^*$ be a string that contains two tandem duplications, namely
\[ x=uvvw=u'v'v'w',\]
for some $u,u',w,w'\in\Sigma^*$ and $v,v'\in\Sigma^+$. We say that the duplication $vv$ \emph{mid-covers} the duplication $v'v'$ in $x$ if
\[
\midp(u'[v'v']w')\in \pos(u[vv]w),
\]
or equivalently, if both of the following requirements are satisfied:
\begin{enumerate}
\item
$\pos(u[vv]w)\cap\pos(u'[v']v'w')\neq\emptyset$
\item
$\pos(u[vv]w)\cap\pos(u'v'[v']w')\neq\emptyset$
\end{enumerate}
\end{definition}

\begin{example}
Consider the alphabet $\Sigma=\set{a,b,c}$, and the string
\[
ababcababcabcaabcbca.
\]
We look at the following tandem duplications, which we highlight by underlining them and placing a vertical line in their middle:
\begin{enumerate}
\item
$a\vv \underline{b\vv a\vv b\vv c\vv a}|\underline{b\vv a\vv b\vv c\vv a}\vv b\vv c\vv a\vv a\vv b\vv c\vv b\vv c\vv a$
\item
$a\vv b\vv a\vv b\vv c\vv a\vv b\vv \underline{a\vv b\vv c} | \underline{a\vv b\vv c}\vv a\vv a\vv b\vv c\vv b\vv c\vv a$
\item
$a\vv b\vv a\vv b\vv c\vv a\vv b\vv a\vv b\vv c\vv a\vv b\vv c\vv a\vv a\vv \underline{b\vv c} | \underline{b\vv c}\vv a$
\end{enumerate}
In this example, the first and third duplications are disjoint, as are the second and third. Additionally, the first duplication mid-covers the second, but the second does not mid-cover the first. We therefore note that the mid-covering relation is not symmetric.
\end{example}

\subsection{Correcting Disjoint Tandem-Duplication Errors}

We first consider a scenario in which codewords are affected by disjoint tandem-duplication errors, and devise a code that corrects them.

We denote the set of all strings obtained from $x\in\Sigma^*$ using any number of disjoint tandem duplications by $D^{\dis,*}(x)$. If we further want to restrict the lengths of each duplicated part to be from some set $L\subseteq\N$ (i.e., in Definition~\ref{def:disjointerror}, $\abs{v_i}\in L$ for all $i$), then we write $D^{\dis,*}_L(x)$. Obviously, $C\subseteq \Sigma^n$ is a code capable of correcting any number of disjoint tandem duplications with lengths from $L$, if and only if for any two distinct $x,y\in C$ we have
\[
D^{\dis,*}_L(x)\cap D^{\dis,*}_L(y)=\emptyset.
\]

Another useful notation is the following: For a set $L\subseteq\N$ we define
\[
L^\Delta \eqdef \set*{\abs*{\ell-\ell'} ~:~ \ell,\ell'\in L, \ell\neq \ell'}.
\]
We can now introduce our code construction for disjoint tandem-duplication errors.

\begin{theorem}
\label{th:disjoint}
Let $L\subseteq\N$ be a given set of duplication lengths, and set $F=L\cup L^\Delta$. Then the code $C_F$, of length $n\in\N$, can correct any number of disjoint tandem duplications with respect to $\cT_L$. Namely,
for all distinct $x,y\in C_F$ we have
\[
D_L^{\dis,*}(x)\cap D_L^{\dis,*}(y)=\emptyset.
\]
\end{theorem}

Before proving this theorem, we require a few lemmas concerning the properties of confusable strings under the theorem's assumptions.

\begin{lemma}
\label{lem:disjointconfusable}
Consider the setting of Theorem~\ref{th:disjoint}. If $x,y\in C_F$ are confusable, then all the following hold:
\begin{enumerate}
\item
We can write
\begin{align*}
x&= x_1 u_1 x_2 u_2 \dots x_t u_t x_{t+1} \\
y&= y_1 v_1 y_2 v_2 \dots y_t v_t y_{t+1}
\end{align*}
where, $t\in\N$, and for all $i$, $x_i,y_i,u_i,v_i\in\Sigma^*$ and $\abs{u_i},\abs{v_i}\in L$.
\item
The following descendants of $x$ and $y$ are equal,
\[
x_1 u_1^2 x_2 u_2^2 \dots x_t u_t^2 x_{t+1}
=
y_1 v_1^2 y_2 v_2^2 \dots y_t v_t^2 y_{t+1}.
\]
Denote this joint descendant as $z$.
\item
For all $1\leq i\leq t$, the duplications $u_i u_i$ and $v_i v_i$ in $z$, mid-cover each other. Additionally, $u_i u_i$ (respectively, $v_i v_i$) does not mid-cover any $v_j v_j$ (respectively, $u_j u_j$) for $i\neq j$.
\end{enumerate}
\end{lemma}

\begin{proof}
By definition, if $x$ and $y$ are confusable, then there exists
\[
z\in D_L^{\dis,*}(x)\cap D_L^{\dis,*}(y).
\]
Since $z$ is obtained from both $x$ and $y$ through a sequence of disjoint tandem duplications whose lengths are from $L$, we may write
\begin{align*}
x&= x_1 u_1 x_2 u_2 \dots x_t u_t x_{t+1} \\
y&= y_1 v_1 y_2 v_2 \dots y_s v_s y_{s+1} \\
z &= x_1 u_1^2 x_2 u_2^2 \dots x_t u_t^2 x_{t+1}
=
y_1 v_1^2 y_2 v_2^2 \dots y_s v_s^2 y_{s+1},
\end{align*}
for some $t,s\in\N$, and for all $i$, $x_i,y_i,u_i,v_i\in\Sigma^*$ and $\abs{u_i},\abs{v_i}\in L$. We note that $s,t\geq 1$ since by definition $x\neq y$ but of the same length $n$, and so each requires at least one duplication to arrive at the same $z$.

Next, we show that each duplication $u_i u_i$ mid-covers some duplication $v_j v_j$ in $z$. Assume to the contrary that some $u_i u_i$ does not mid-cover any $v_j v_j$ in $z$. Thus, for each $1\leq j\leq s$,
\[
\pos(\dots[u_i u_i]\dots) \cap \pos(\dots[v_j]v_j\dots)=\emptyset
\]
or
\[
\pos(\dots[u_i u_i]\dots) \cap \pos(\dots v_j[v_j]\dots)=\emptyset
\]
where the dots should be completed to obtain $z$. Now, start with $z$, and for each $1\leq j\leq s$, take $v_j v_j$ and erase one copy of $v_j$ whose positions do not intersect those of $u_i u_i$ (if both copies satisfy this property, pick one arbitrarily). The resulting string is
\[ y_1 v_1 y_2 v_2 \dots y_s v_s y_{s+1} = y.\]
However, none of the positions of $u_i u_i$ were erased, and so $u_i u_i$ is a factor of $y$. Since $\abs{u_i}\in L\subseteq F$, we have that $y\in C_F$ but also contains a factor that is a tandem duplication of length in $F$. This contradicts Definition~\ref{def:cf}. It follows that each duplication $u_i u_i$ mid-covers some $v_j v_j$ in $z$. A symmetric argument shows each duplication $v_j v_j$ mid-covers some $u_i u_i$ in $z$.

Assume the duplication $u_i u_i$ mid-covers $v_j v_j$ in $z$. We note that $u_{i'} u_{i'}$, with $i\neq i'$, cannot also mid-cover the same $v_j v_j$ in $z$. This follows from:
\begin{align*}
&\midp(\dots[v_j v_j]\dots)\in\pos(\dots[u_i u_i]\dots) && \text{($u_i u_i$ mid-covers $v_j v_j$)}\\
&\pos(\dots[u_i u_i]\dots)\cap\pos(\dots[u_{i'}u_{i'}]\dots)=\emptyset && \text{($u_i u_i$ and $u_{i'}u_{i'}$ are disjoint)}
\end{align*}

and so
\[
\midp(\dots[v_j v_j]\dots)\not\in\pos(\dots[u_{i'} u_{i'}]\dots).
\]
Since each $u_i u_i$ mid-covers some $v_j v_j$ in $z$, and distinct $u_i u_i$ mid-cover distinct $v_j v_j$ in $z$, we have $t\leq s$. A symmetric argument gives $s\leq t$, implying $s=t$.

Combining all of the above, each $u_i u_i$ mid-covers a unique $v_j v_j$ in $z$, and vice versa. Consider the duplication $u_1 u_1$. Assume to the contrary it mid-covers $v_j v_j$ for some $j\geq 2$. This will leave $v_1 v_1$ not mid-covered by any $u_i u_i$, $i\geq 2$, since these appear to the right of $u_1 u_1$ in $z$ whereas $v_1 v_1$ appears to the left of $v_j v_j$, and the duplications are disjoint. Hence, $u_1 u_1$ must cover $v_1 v_1$ in $z$, and symmetrically, $v_1 v_1$ mid-covers $u_1 u_1$ in $z$. Iterating over this argument proves that $u_i u_i$ and $v_i v_i$ mid-cover each other, for all $1\leq i\leq t$, which completes the proof.
\end{proof}

\begin{lemma}
\label{lem:eqmidcover}
Assume a string $z\in\Sigma^*$ can be written in two ways,
\[
z = uvvw = u'v'v'w',
\]
where $u,v,w,u',v',w'\in\Sigma^*$, and $\abs{v}=\abs{v'}>0$. If $vv$ and $v'v'$ mid-cover each other in $z$, then
\[ uvw=u'v'w',\]
namely, removing the duplications results in the same string.
\end{lemma}

\begin{proof}
Assume that $z$ satisfies the two decompositions into factors, as above. If $\abs{u}=\abs{u'}$, then $u=u'$, $v=v'$, $w=w'$, and the claim trivially follows. Let us now assume, w.l.o.g., that $\abs{u}<\abs{u'}$. Define $\delta\eqdef\abs{u'}-\abs{u}$. By our assumption that $vv$ and $v'v'$ mid-cover each other in $z$, we have $0<\delta<\abs{v}$.

We decompose $z$ into the following factors (consult Fig.~\ref{fig:eqmidcover} throughout the rest of the proof),
\[ z = z_1 z_2 z_3 z_4 z_5 z_6 z_7 \]
where
\begin{align*}
\abs*{z_1} & = \abs{u} & \abs*{z_2} & = \delta & \abs*{z_3} &= \abs*{v}-\delta & \abs*{z_4} &=\delta \\
\abs*{z_5} &= \abs*{v}-\delta & \abs*{z_6}&=\delta & \abs*{z_7}&=\abs*{w}-\delta.
\end{align*}
Thus,
\begin{align*}
u &= z_1 & v &= z_2 z_3 = z_4 z_5 & w & = z_6 z_7 \\
u' &= z_1 z_2 & v'&= z_3 z_4 = z_5 z_6 & w' &= z_7.
\end{align*}

\begin{figure}
\begin{center}
\begin{overpic}
[width=\textwidth]
{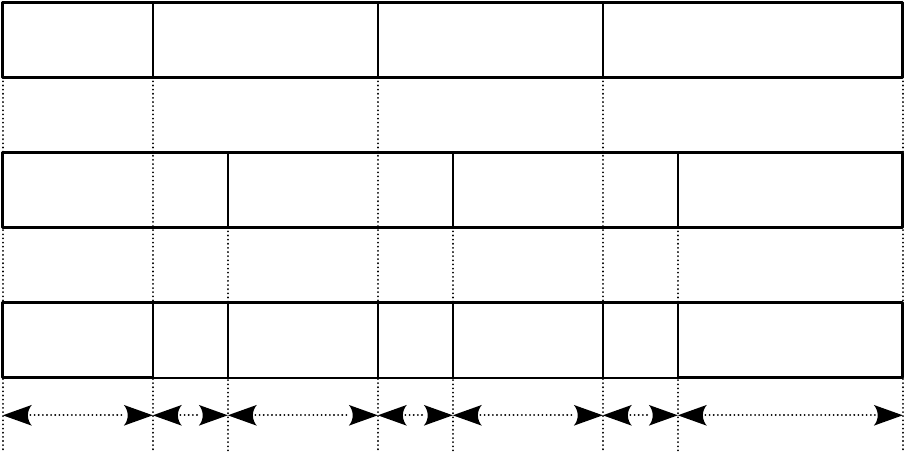}
\put(8,45){$u$}
\put(28,45){$v$}
\put(53,45){$v$}
\put(83,45){$w$}
\put(13,28.5){$u'$}
\put(37,28.5){$v'$}
\put(62,28.5){$v'$}
\put(86,28.5){$w'$}
\put(8,12){$z_1$}
\put(20,12){$z_2$}
\put(32,12){$z_3$}
\put(44,12){$z_4$}
\put(56,12){$z_5$}
\put(69,12){$z_6$}
\put(86,12){$z_7$}
\put(7,0){$\abs{u}$}
\put(20,0){$\delta$}
\put(29.5,0){$\abs{v}-\delta$}
\put(45,0){$\delta$}
\put(54,0){$\abs{v}-\delta$}
\put(70,0){$\delta$}
\put(83,0){$\abs{w}-\delta$}
\end{overpic}
\end{center}
\caption{The decomposition of $z$ into factors in the proof of Lemma~\ref{lem:eqmidcover}.}
\label{fig:eqmidcover}
\end{figure}
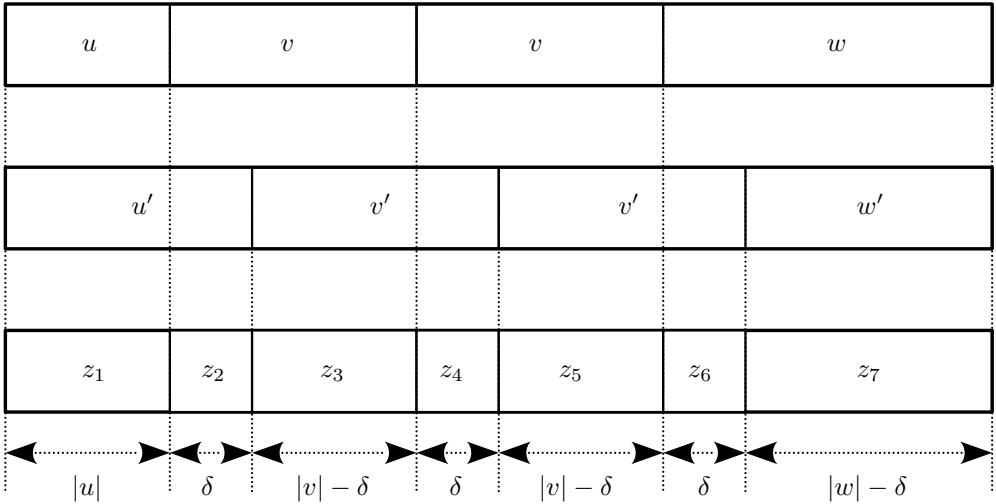

We now have
\[
z_2 = z_4 = z_6,
\]
where the first equality is due to the $vv$ duplication, and the second equality is due to the $v'v'$ duplication. We also have
\[
z_3=z_5,
\]
due to either the $vv$ or the $v'v'$ duplications. It follows that
\[
uvw = z_1 z_2 z_3 z_2 z_7 = u' v' w',
\]
which completes the proof.
\end{proof}

\begin{lemma}
\label{lem:neqmidcover}
Assume a string $z\in\Sigma^*$ can be written in two ways,
\[
z = uvvw = u'v'v'w',
\]
where $u,v,w,u',v',w'\in\Sigma^*$, and $\abs{v}>\abs{v'}>0$. If $vv$ and $v'v'$ mid-cover each other in $z$, then $u'v'w'$ contains a factor $v''v''$, where $\abs{v''}=\abs{v}-\abs{v'}$.
\end{lemma}

\begin{proof}
Assume $z$ satisfies the two decompositions into factors, as above. We further assume that the mid-point of the $v'v'$ duplication is not to the left of the mid-point of the $vv$ duplication. Namely, we assume that
\begin{equation}
\label{eq:rightmid}
\abs{u}+\abs{v}\leq \abs{u'}+\abs{v'}
\end{equation}
(see Fig.~\ref{fig:neqmidcover} throughout the rest of the proof). The proof of the symmetric case proceeds as a mirror image of the rest of this proof.

\begin{figure}
\begin{center}
\begin{overpic}
[width=\textwidth]
{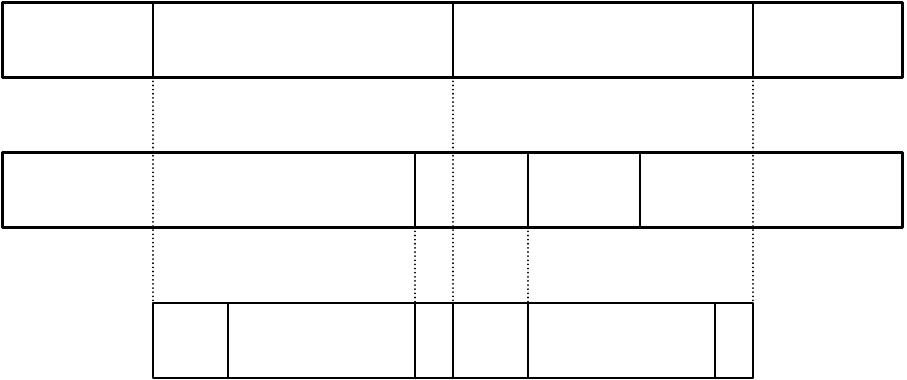}
\put(8,37){$u$}
\put(33,37){$v$}
\put(66,37){$v$}
\put(91,37){$w$}
\put(23,20){$u'$}
\put(51,20){$v'$}
\put(64,20){$v'$}
\put(84,20){$w'$}
\put(20.5,3.5){$b$}
\put(35,3.5){$v''$}
\put(47,3.5){$a$}
\put(53.5,3.5){$b$}
\put(68,3.5){$v''$}
\put(80,3.5){$a$}
\end{overpic}
\end{center}
\caption{The decomposition of $z$ into factors in the proof of Lemma~\ref{lem:neqmidcover}.}
\label{fig:neqmidcover}
\end{figure}
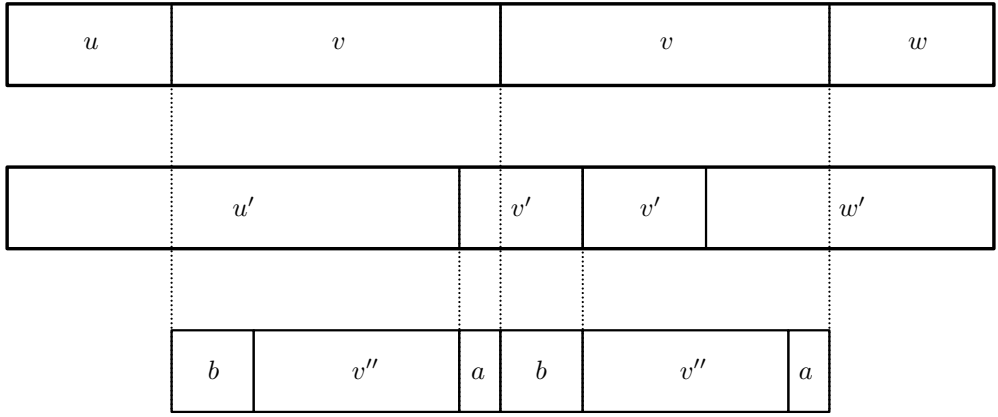

We further decompose $v'$ according to the mid-point of the $vv$ duplication with relation to the first copy of $v'$. More precisely, we write
\[ v' = ab,\]
where $\abs{a}=\abs{u}+\abs{v}-\abs{u'}$. The latter is positive since $v'v'$ mid-covers $vv$. We then have $\abs{b}=\abs{u'}+\abs{v'}-\abs{u}-\abs{v}$, which is non-negative by~\eqref{eq:rightmid}.

We now have that $a$ is a suffix of the left copy of $v$. We also have that $b$ is a prefix of the right copy of $v$. It must therefore be a prefix of the left copy of $v$ as well. Looking at the left copy of $v$ we can write
\[ v = bv''a,\]
where $\abs{v''}=\abs{v}-\abs{v'}$, which is positive by the assumptions of the theorem. Again, the right copy of $v$ must also equal $bv''a$.

Our final observation is that $v''$ is both a suffix of $u'$, and a prefix of $v'w'$. Thus, the string $u'v'w'$ contains the factor $v''v''$ as claimed.
\end{proof}

We are now in a position to prove Theorem~\ref{th:disjoint}.

\begin{proof}[Proof of Theorem~\ref{th:disjoint}]
Assume to the contrary that $C_F$ is not a code as desired. Hence, there exist two (distinct) confusable strings, $x,y\in C_F$, with the properties described in Lemma~\ref{lem:disjointconfusable}. As in the notation of Lemma~\ref{lem:disjointconfusable}, let $u_i u_i$ be the tandem duplications corrupting $x$, and $v_i v_i$ be those corrupting $y$, for all $1\leq i\leq t$.

Assume, first, that there exists $1\leq i\leq t$, such that $\abs{u_i}\neq \abs{v_i}$. W.l.o.g., assume $\abs{u_i}>\abs{v_i}$. We take $z$ and remove from it one copy of $v_i$, resulting in a string $z'$. By Lemma~\ref{lem:neqmidcover}, $z'$ contains a factor $v'' v''$. Furthermore, $\abs{v''}\in L^\Delta\subseteq F$. We also observe that, by the proof of Lemma~\ref{lem:neqmidcover}, the factor $v'' v''$ results only from positions contained in $\pos(\dots[u_i u_i]\dots)$ in $z$. Since by Lemma~\ref{lem:disjointconfusable}, $u_i u_i$ does not mid-cover any $v_{i'}v_{i'}$ for $i\neq i'$, we can further remove from $z'$ one copy of $v_{i'}$, for all $1\leq i'\leq t$, $i'\neq i$, in positions disjoint from $\pos(\dots[u_i u_i]\dots)$ in $z$. The resulting string is just $y$, with the $v'' v''$ factor remaining intact. Since $\abs{v''}\in F$, we have a contradiction to $y\in C_F$ being a codeword.

We assumed $\abs{u_i}\neq \abs{v_i}$ for some $i$, and reached a contradiction. Hence, to avoid a contradiction, for all $1\leq i\leq t$ we must have $\abs{u_i}=\abs{v_i}$. However now, by repeatedly using Lemma~\ref{lem:eqmidcover}, removing the duplicated parts $u_i$ (resulting in $x$) or removing the duplicated parts $v_i$ (resulting in $y$) leaves us with the same string, namely, $x=y$. But that contradicts $x\neq y$ required for $x$ and $y$ to be confusable. Hence, $C_F$ does not contain a pair of confusable codewords, and is therefore capable of correcting any number of disjoint tandem duplications with lengths from $L$.
\end{proof}

We remark on two specific cases in particular. First, assume we are interested in correcting tandem duplications of length $\ell\in\N$ only. This case, called \emph{uniform tandem duplication}, has been first studied in~\cite{JaiFarSchBru17a}. Up to a slight modification (that does not affect the asymptotic code rate), the code constructed in~\cite[Theorem 15]{JaiFarSchBru17a} is $C_{\set{\ell}}$. In the setting of this paper, to correct disjoint uniform tandem duplications of length $\ell$, we set $L=\set{\ell}$, and then $L^\Delta=\emptyset$, giving us $F=\set{\ell}$, resulting in the same code $C_{\set{\ell}}$. While Theorem~\ref{th:disjoint} ends up with the same code as~\cite{JaiFarSchBru17a}, unfortunately in this case, it only guarantees the correction of \emph{disjoint} duplications.

Our second remark concerns $L$ of the form $[\ell]\eqdef\set{1,2,\dots,\ell}$. The cases of $L=[2]$ and $L=[3]$ were studied in~\cite{JaiFarSchBru17a}, where it was proved that choosing $F=L$ (namely, constructing $C_{[2]}$ and $C_{[3]}$) results in error-correcting codes capable of correcting any number of tandem duplications with lengths from $[2]$ or $[3]$ respectively. Extending these results to larger $k$ seems difficult, with~\cite{ZerEsmGul19} providing codes for $[\ell]$ with $4\leq \ell\leq 9$. These codes, however, are not of the form $C_F$. Using Theorem~\ref{th:disjoint}, we note that $[\ell]\cup [\ell]^{\Delta}=[\ell]$. Thus, choosing $F=L$ allows us to extend the results of~\cite{JaiFarSchBru17a} to any $\ell\in\N$, however, at the price of correcting only disjoint tandem duplications.

\subsection{Correcting Equal-Length Tandem-Duplication Errors}
   	
The second scenario we study has strings that are corrupted by any number of tandem duplications (not necessarily disjoint), with lengths from a given set $L$, but with all duplications of equal length.

Formally, let $x\in\Sigma^*$ be a string. We recall the definition of the tandem-duplication rules from~\eqref{eq:tandem}. Let $L\subseteq\N$ be a set of allowed duplication lengths. We say $z\in\Sigma^*$ is obtained from $x$ using $t$ duplications of equal length, if there exists $\ell\in L$, and non-negative integers $i_1,i_2,\dots,i_t$, such that
\[ z = T_{i_t,\ell}(T_{i_{t-1},\ell}(\dots T_{i_1,\ell}(u)\dots)),\]
and we denote it by $x\der^{=,t}_\ell z$. We denote the set of all strings obtained from $x$ by using any number of equal-length tandem duplications with lengths from $L$, by $D^{=,*}_L(x)$. Similar to the previous section, $C\subseteq \Sigma^n$ is a code capable of correcting any number of equal-length tandem duplications with lengths from $L$, if and only if for any two distinct $x,y\in C$ we have
\[
D^{=,*}_L(x)\cap D^{=,*}_L(y)=\emptyset.
\]

The code construction now follows:

\begin{theorem}
\label{th:eq}
Let $L\subseteq\N$ be a given set of duplication lengths, and let $F=L$. We further assume that for any $\ell,\ell'\in L$, $\ell>\ell'$, we have $\ell\geq 2\ell'$. Then the code $C_F$ of length $n\in\N$ can correct any number of equal-length tandem duplications with respect to $\cT_L$. Namely, for all distinct $x,y\in C_F$ we have
\[
D^{=,*}_L(x)\cap D^{=,*}_L(y)=\emptyset.
\]
\end{theorem}

To prove this theorem we need a simple lemma, whose proof requires a tool first described in~\cite{FarSchBru16}. Denote the size of the alphabet by $q\eqdef\abs{\Sigma}$, $q\geq 2$, and w.l.o.g., assume $\Sigma=\Z_q$, the cyclic Abelian group of integers with addition modulo $q$. For $\ell\in\N$, we define $\phi_\ell:\Sigma^*\to\Sigma^*$ by
\[
\phi_\ell(x) = \pref_{\abs{x}}(x 0^\ell - 0^\ell x )
\]
where we assume $\abs{x}\geq \ell$, subtraction is performed element-wise over $\Z_q$, and where $\pref_m(z)$ denotes the $m$-prefix of $z$. It is easy to see that $\phi_\ell$ is invertible, and that the $\ell$-prefix of $\phi_\ell(x)$ is the $\ell$-prefix of $x$. Except for the $\ell$-prefix of $\phi_\ell(x)$, the remaining string contains $0^\ell$ as a factor exactly in positions in which an $\ell$-length tandem duplication occurs (see ~\cite[Lemma 6]{JaiFarSchBru17a}). Moreover, if $t$ tandem duplications of length $\ell$ occurred, we can find $t$ pair-wise disjoint $\ell$-factors whose de-duplication will restore the original string. Here, by de-duplication we mean replacing a factor $u^2$ by $u$.

\begin{example}
\label{ex:phi1}
Consider $\Sigma=\Z_6$, and let $\ell=2$. Assume $x=054213$, so then
\[
\phi_2(x) = \pref_{\abs{x}}(05421300-00054213)=
\pref_6(05433153)=054331.
\]
and consider the following sequence of tandem duplications of length $2$ (where the tandem duplication is underlined, and the middle is a vertical line):
\[
x=054213 \der 0\underline{54}|\underline{54}213
\der 054542\underline{13}|\underline{13}
\der 05\underline{45}|\underline{45}421313=z.
\]
Applying the $\phi_2$ transform, this derivation looks like:
\[
\phi_2(x)=054331 \rightarrow 054\underline{00}331 \rightarrow 05400331\underline{00} \rightarrow 0540\underline{00}033100 = \phi_2(z),
\]
where we underlined the inserted factors $0^2$ due to the tandem duplications. We further note three disjoint $0^2$,
\[
\phi_2(z) = 054\underline{00}\,\underline{00}331\underline{00}
\]
corresponding to three $2$-factors in $z$,
\[z = 054\underline{54}\,\underline{54}213\underline{13}\]
whose removal restores the original string $x$.
\end{example}

\begin{lemma}
\label{lem:phi}
Let $\ell,t\in\N$, and let $x\in\Sigma^*$ be a string, $\abs{x}\geq \ell$. If $x\der^{=,t}_\ell z$, then $z$ contains at least $\ceil{t/2}$ disjoint tandem duplications of length $\ell$.
\end{lemma}

\begin{proof}
Recall that we assume w.l.o.g. that $\Sigma=\Z_q$. Consider $\phi_\ell(x)$ and $\phi_\ell(z)$, ignoring their $\ell$-prefix. Since $z$ is obtained from $x$ using $t$ tandem duplications of length $\ell$, by~\cite[Lemma 6]{JaiFarSchBru17a}, $\phi_\ell(z)$ is obtained from $\phi_\ell(x)$ by inserting a factor $0^\ell$ a total of $t$ times.

We run the following simple procedure on $\phi_\ell(z)$ (without its $\ell$-prefix): We scan $\phi_\ell(z)$ from left to right, finding the first factor $0^\ell$, declaring it a tandem duplication of length $\ell$. We skip the $\ell$ positions following it (regardless of their content), and repeat. Since at the worst case, the $\ell$ skipped positions are another $0^\ell$ factor, we are guaranteed to declare at least $\ceil{t/2}$ tandem duplications. Since their locations are separated by at least $\ell$ letters, they must be pair-wise disjoint tandem duplications.
\end{proof}

\begin{example}
Continuing Example~\ref{ex:phi1}, in which $t=3$, Lemma~\ref{lem:phi} guarantees the existence of $\ceil{3/2}=2$ disjoint tandem duplications in $z$. The scanning process of $\phi_2(z)$ declares the following two places (underlined):
\[
\phi_2(z) = 054\underline{00}00331\underline{00},
\]
which indeed correspond to the following two disjoint tandem duplications in $z$:
\[
z=0\underline{54}|\underline{54}542\underline{13}|\underline{13}.
\]
\end{example}

\begin{proof}[Proof of Theorem~\ref{th:eq}]
Assume to the contrary that $C_F$ is not a code as desired. Hence, there exist two (distinct) confusable strings, $x,y\in C_F$. Since $x$ and $y$ are confusable, there exists
\[
z\in D^{=,*}_L(x) \cap D^{=,*}_L(y).
\]
It follows that there exist duplication lengths $\ell_x,\ell_y\in L$, and $t_x,t_y\in\N$, such that
\[
x \der^{=,t_x}_{\ell_x} z \qquad\text{and}\qquad
y \der^{=,t_y}_{\ell_y} z.
\]
In other words, we may reach $z$ either by $t_x$ duplications of length $\ell_x$ acting on $x$, or by $t_y$ duplications of length $\ell_y$ acting on $y$. Since in both cases we start with strings of length $\abs{x}=\abs{y}=n$, and end with the same string $z$, we necessarily have
\begin{equation}
\label{eq:prod}
t_x \ell_x = t_y \ell_y.
\end{equation}

We contend that we cannot have $\ell_x=\ell_y$. Assume to the contrary that $\ell\eqdef \ell_x=\ell_y$. By~\cite[Theorem 15]{JaiFarSchBru17a}, the code $C_{\set{\ell}}$ can correct any number of tandem duplications of length $\ell$. Since $\set{\ell}\subseteq F$, we have $C_F\subseteq C_{\set{\ell}}$, meaning $C_F$ can also correct any number of tandem duplications of length $\ell$. In particular, it does not contain two confusable codewords, a contradiction. Thus, $\ell_x\neq \ell_y$, and assume w.l.o.g. that $\ell_x>\ell_y$, and therefore, by~\eqref{eq:prod},
\[t_x<t_y.\]

By Lemma~\ref{lem:phi}, $z$ contains $\ceil{t_y/2}$ disjoint tandem duplications of length $\ell_y$. We can therefore write,
\[
z = y_1 v_1^2 y_2 v_2^2 y_3 \dots y_{\ceil{t_y/2}} v_{\ceil{t_y/2}}^2 y_{\ceil{t_y/2}+1},
\]
where $y_i\in\Sigma^*$ and $v_i\in\Sigma^{\ell_y}$, for all $i$.

Recall that $z$ was obtained from $x$ using $t_x$ tandem duplications of length $\ell_x$. Thus, $z$ must contain $t_x$ factors, $u_i^2$, $\abs{u_i}=\ell_x$, $1\leq i\leq t_x$, that are not necessarily disjoint, and whose de-duplication restores $x$ (see Example~\ref{ex:phi1}, and the discussion before it). W.l.o.g., we assume they are indexed from left to right. In particular, if we define
\[
m_i\eqdef \midp(\dots[u_i u_i]\dots),
\]
then
\[ m_1 < m_2 < \dots < m_{t_x}.\]

In a similar fashion to the previous section, we contend that each duplication $v_j v_j$ must mid-cover at least one duplication $u_i u_i$ in $z$. If we assume to the contrary that there is a duplication $v_j v_j$ that does not mid-cover any duplication $u_i u_i$, then one of the following intervals,
\[
(-\infty,m_1),(m_1,m_2),(m_2,m_3),\dots,(m_{t_x-1},m_{t_x}),(m_{t_x},\infty)
\]
entirely contains $\pos(\dots[v_j v_j]\dots)$. However, then, we can de-duplicate all the  $u_i u_i$ from $z$ (thus, obtaining $x$) but with the $v_j v_j$ factor remaining intact. Hence, $x$ contains a tandem duplication of length $\ell_y\in F$, contradicting the fact that $x\in C_F$.

Recall that the duplications $v_j v_j$, $1\leq j\leq \ceil{t_y/2}$, are pair-wise disjoint. Thus, a duplication $u_i u_i$ cannot be mid-covered by more than one duplication $v_j v_j$. Hence, $\ceil{t_y/2}\leq t_x$, implying
\begin{equation}
\label{eq:firstup}
t_y \leq 2t_x.
\end{equation}

Next, we focus on the extreme case of $t_y=2t_x$, and therefore, $\ell_x=2\ell_y$. By the above, it follows that $v_i v_i$ mid-covers $u_i u_i$ in $z$ for all $1\leq i\leq t_x$. It shall suffice to look at $v_1 v_1$ which mid-covers $u_1 u_1$. Looking at the positions $u_1 u_1$ occupies in $z$, we can write the equation
\begin{equation}
\label{eq:twice}
 u_1 u_1 = a v_1 v_1 b,
\end{equation}
with $1<\abs{a}<2\ell_y$, to ensure that $v_1 v_1$ mid-cover $u_1 u_1$. We assume that the middle of $v_1 v_1$ is not to the left of the middle of $u_1 u_1$ (see Fig.~\ref{fig:twice} throughout the rest of the proof), namely, $\ell_y\leq \abs{a}<2\ell_y$. The other case of $1 < \abs{a}<\ell_y$ has a symmetric proof.

\begin{figure}
\begin{center}
\begin{overpic}
[scale=0.75]
{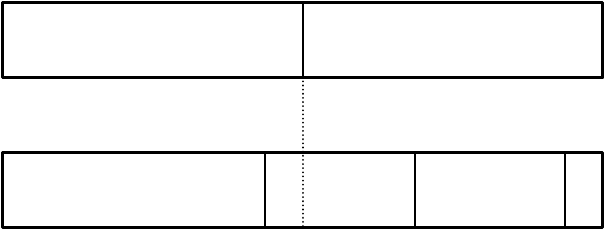}
\put(24,31){$u_1$}
\put(74,31){$u_1$}
\put(20,6){$a$}
\put(55,6){$v_1$}
\put(80,6){$v_1$}
\put(95.5,6){$b$}
\end{overpic}
\end{center}
\caption{The two decomposition of~\eqref{eq:twice}.}
\label{fig:twice}
\end{figure}
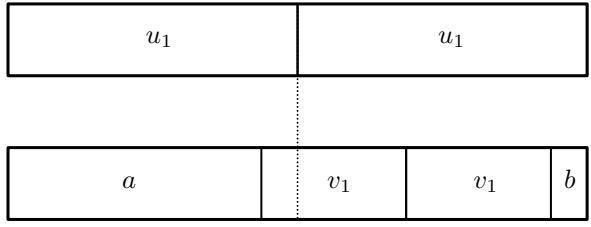

Looking at the left copy of $v_1$, we further decompose it according to the midpoint of $u_1 u_1$. Namely, we write
\[
v_1 = cd
\]
where
\[ \abs{c} = \ell_x-\abs{a}=2\ell_y-\abs{a}, \qquad \abs{d} = \ell_y-\abs{c} = \abs{a}-\ell_y.\]
We observe that
\[
\abs{b} = 2\ell_x-2\ell_y-\abs{a} = \abs{c}.
\]
By the $u_1 u_1$ duplication we therefore have $c=b$, and therefore,
\[ v_1 = cd. \]
Now, looking at the second copy of $u_1$ we have
\[
u_1 = d v_1 b = dcdc.
\]
Since $u_1$ is a factor of the codeword $x$, we observe that $x$ contains $(dc)^2$ as a factor, namely a tandem duplication of length $\ell_y\in F$, which is a contradiction.

Having ruled out the extreme case of $t_y=2t_x$, by~\eqref{eq:firstup}, we must have
\[
t_y < 2t_x.
\]
But that means
\[
2 > \frac{t_y}{t_x} \overset{(a)}{=} \frac{\ell_x}{\ell_y} \overset{(b)}{\geq} 2,
\]
where $(a)$ follows from~\eqref{eq:prod}, and $(b)$ follows from $\ell_x>\ell_y$ and the theorem requirement of the set $L$. We have reached our final contradiction, and so no two confusable codewords exist in $C_F$.
\end{proof}

We briefly provide two remarks concerning the construction from Theorem~\ref{th:eq}. First, we note that when $L=\set{\ell}$ is a singleton, it satisfies (vacuously) the conditions of Theorem~\ref{th:eq}. Thus, $C_{\set{\ell}}$ from~\cite{JaiFarSchBru17a} is a special case of Theorem~\ref{th:eq}, as is also the ad-hoc case of $L=\set{1,2}$ appearing in~\cite{JaiFarSchBru17a}.

Our second remark concerns the efficient decoding of the code from Theorem~\ref{th:eq}. This decoding relies on the linear-time decoding of the code from~\cite[Theorem 15]{JaiFarSchBru17a} which we now describe. Assuming $x\in\Sigma^n$ was stored, and was corrupted by tandem duplications of length $\ell$, resulting in a received string $z\in\Sigma^N$. We further assume the alphabet used is $\Sigma=\Z_q$, and we may therefore compute $\phi_\ell(z)$ and write
\[
\phi_\ell(z) = z_0 0^{m_1} z_1 0^{m_2} z_2 \dots 0^{m_k} z_k 0^{m_{k+1}},
\]
where $m_i$ are non-negative integers, $\abs{z_0}=\ell$, and $z_i$, $i\geq 1$, do not contain the letter $0$. It was proved in~\cite{JaiFarSchBru17a}, that de-duplicating all the tandem duplications of length $\ell$ results in the original codeword $x$. Since such duplications are equivalent to factors $0^\ell$ in $\phi_\ell(z)$, this decoding procedure is simply
\begin{equation}
\label{eq:decode}
x = \phi^{-1}_\ell\parenv*{
z_0 0^{m_1 \bmod \ell} z_1 0^{m_2 \bmod \ell} z_2 \dots 0^{m_k \bmod\ell } z_k 0^{m_{k+1} \bmod \ell}
}.
\end{equation}
Thus, computing $x$ requires $O(N)$ time, namely, time that is linear in the length of the received string. In the case of the code $C_F$ from Theorem~\ref{th:eq}, we do not know $\ell$. We do, however, know the length of the transmitted string, $n$, and so, necessarily, $\ell\leq n$ and $\ell\in L$. We can now try all possible $\ell\in L$, $\ell\leq n$, and see whether~\eqref{eq:decode} results in a decoded string of length $n$. By Theorem~\ref{th:eq}, we are guaranteed only one choice of $\ell$ is possible. The decoding complexity is $O(nN)=O(N^2)$.

To conclude this section, we can also combine the error models of Theorem~\ref{th:disjoint} and Theorem~\ref{th:eq}.

\begin{theorem}
\label{th:disjointeq}
Let $L\subseteq\N$ be a given set of duplication lengths, and let $F=L$. Then the code $C_F$ of length $n\in\N$ can correct any number of disjoint equal-length tandem-duplication errors with respect to $\cT_L$. Namely, for all distinct $x,y\in C_F$ we have
\[
D^{\dis,=,*}_L(x)\cap D^{\dis,=,*}_L(y)=\emptyset.
\]
\end{theorem}

\begin{proof}
The proof is essentially a simple modification of the proof of Theorem~\ref{th:eq}. Only now, we are guaranteed there exist $t_y$ disjoint tandem duplication of length $\ell_y$, and so~\eqref{eq:firstup} becomes
\[
t_y\leq t_x,
\]
and we immediately reach a contradiction.
\end{proof}

\section{Conclusion}
\label{sec:conc}

In this paper we considered duplication-free codes. We provided two new main parametric constructions, for the case of disjoint duplication errors, and the case of equal-length duplication errors. We then also constructed a code for the combination of the two scenarios. All the code constructions are given a set $L$ of duplication lengths occurring in the channel, and design a set $F$ of forbidden duplications in the codewords. All of the constructed codes have positive asymptotic rate.

Several open questions remain. In particular, it is unknown whether the constructed codes in this paper attain the maximum asymptotic coding rate, given $L$. Another open question concerns the construction of Theorem~\ref{th:disjoint}. Unlike the second construction, it is lacking an efficient decoding algorithm. We leave these questions and others, for future work.

\backmatter

\bmhead{Acknowledgments}

This work was supported in part by the Zhejiang Lab BioBit Program (grant
no. 2022YFB507). The author M.~Schwartz is currently on a leave of absence from Ben-Gurion University of the Negev.

\bibliography{allbib}

\end{document}